\documentclass[conference]{IEEEtran}
\usepackage{mathtools, cuted} 
\usepackage{lettrine} 
\usepackage{amsthm,amsmath} 
\usepackage{epsfig,amssymb,amsbsy,verbatim,subfigure,array}
\usepackage{pstricks,cite,url}
\usepackage{multicol,afterpage,wrapfig,paralist} 
\usepackage{hyperref}
\usepackage{breakurl}
\usepackage{verbatim,tikz,subfigure}
\usepackage{tikz,pgfplots}
\usetikzlibrary{patterns} 
\usepackage[english]{babel}
 
\newtheorem{lemma}{Lemma}{}
  
  \newtheorem{theorem}{Theorem}

\def\d{\mathrm{d}}

\title{Can Balloons Produce Li-Fi?\\A Disaster Management Perspective}
\author{\IEEEauthorblockN{Atchutananda Surampudi$^{1}$, Sankalp Sirish Chapalgaonkar$^{1}$ and Paventhan Arumugam$^{2}$}  
\IEEEauthorblockA{$^{1}$Department of Electrical Engineering, Indian Institute of Technology Madras,\\
$^{2}$ERNET, IIT Madras Research Park.\\
$^{1}$\{ee16s003, ee15b018\}@ee.iitm.ac.in, $^{2}$paventhan@eis.ernet.in}
}

\begin{document}
\maketitle

\begin{abstract}
Natural calamities and disasters disrupt the conventional communication setups and the wireless bandwidth becomes constrained. A safe and cost effective solution for communication and data access in such scenarios is long needed. Light-Fidelity (Li-Fi) which promises wireless access of data at high speeds using visible light can be a good option. Visible light being safe to use for wireless access in such affected environments, also provides illumination. Importantly, when a Li-Fi unit is attached to an air balloon and a network of such Li-Fi balloons are coordinated to form a Li-Fi balloon network, data can be accessed anytime and anywhere required and hence many lives can be tracked and saved. We propose this idea of a Li-Fi balloon and give an overview of it's design using the Philips Li-Fi hardware. Further, we propose the concept of a balloon network and coin it with an acronym, the LiBNet. We consider the balloons to be arranged as a homogenous Poisson point process in the LiBNet and we derive the mean co-channel interference for such an arrangement.
\end{abstract}

\section{Introduction}
Visible-light-communications (VLC) or Light-Fidelity (Li-Fi) has gained a lot of prominence among the research groups around the world for it's secure, safe, energy efficient and high speed data transfer characteristics over a wireless medium. The use of visible light, which has given life to the planet, is now paving the path for a new form of wireless communications \cite{haraldhass}. This technology has both the characteristics of providing data access and illumination at the same time over linear-time-invariant optical wireless channels \cite{barrykahn1}, \cite{barrykahn2}. \\ \par 

Natural disasters like earthquakes or floods, disrupt the available wireless communication infrastructure and as a result the wireless bandwidth becomes scarce. But, in such scenarios, communication becomes an important aspect of the emergency response. A safe, high speed, energy efficient and spectrum efficient wireless access technology is the need of the hour. There have been several technologies in use for emergency communications, like the ham radio \cite{ham}, cognitive radio \cite{cr1}, \cite{cr2} etc. But, when the demand becomes large or large number of people have to be saved, these technologies remain sub-optimal due to the limited radio frequency (RF) spectrum available in such a situation. Also, if floods are considered, the people trapped below flood waters have to be saved immediately. Communicating with them, or to at least receive their location using some uplink signatures becomes a priority. RF waves cannot penetrate through flood waters, whereas the light wave can. So, Li-Fi can be a reliable and energy efficient option to improve the communication capacity along with illuminating the affected area. Lately, for practical purposes and outdoor applications, Li-Fi has been implemented using fixed light-emitting-diode (LED) sources \cite{haasprof}. The question that arises is, can these Li-Fi LEDs fly in the air, without any constraints, at a given height near the ground at around $10$ to $15$ metres, to provide the same data access whenever and wherever required? This idea of a mobile Li-Fi unit can be a good solution during such natural calamities. This can be achieved by using a bunch of centrally monitored balloons, which can carry the Li-Fi LEDs operated by low power batteries to illuminate as well as provide data access (or broadcast services) to a given region for a given period of time.  \\ \par

\subsection{Our approach and contributions}
Our contributions are as follows.
\begin{itemize}
\item \textit{The Li-Fi Balloon - } We present an overview of the design for the balloon using the Philips Li-Fi hardware. 
\item The concept of LiBNet and the overall functioning.  
\item The mean co-channel interference is derived for homogenous Poisson arrangement of balloons in one and two dimensions.
\end{itemize}

\subsection{Context of application:} 
\subsubsection{Earthquakes}During earthquakes, every other communication infrastructure above or below the ground gets disrupted. Also, the problem deepens when any gas leaks happen and become susceptible to combustion. So, using the LiBNet, both illumination and high speed data services can be provided, in a safe way, without interfering with the RF frequencies. \\ \par

\subsubsection{Floods}In a flood affected region, LiBNet will be able to send a broadcast inside the water surface and receive any uplink messages from the people or pets trapped below. The balloons can coordinate to provide tracking services as well.\\ \par         

The further organisation of the paper is as follows. In section II, we give a brief overview of the design of the Li-Fi balloon using the hardware from Philips. In section III, we present the concept of LiBNet and assuming a poisson point process arrangement, we derive the mean co-channel interference in such a scenario. The paper concludes with section IV. 

\section{The Li-Fi Balloon}
The Li-Fi balloon can be designed with the outline presented in Fig. \ref{overview}.
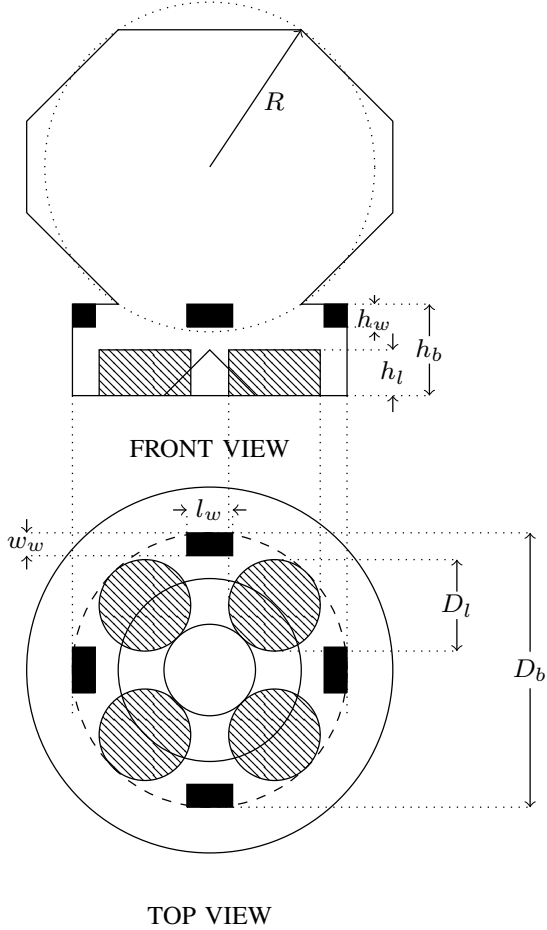
\begin{figure}[ht]
\centering
\resizebox{0.42\textwidth}{!}{%
 \begin{tikzpicture}
 \draw [dotted] (2.5,8.5) circle [radius=1.802];
 \draw (0.5,8) -- (0.5,9) -- (1.5,10) -- (3.5,10) -- (4.5,9) -- (4.5,8) -- (3.5,7) -- (4,7) -- (4,6) -- (1,6) -- (1,7) -- (1.5,7) -- (0.5,8);
 \draw [pattern=north west lines] (1.292,6) -- (1.292,6.5) -- (2.293,6.5) -- (2.293,6);
 \draw (2,6) -- (2.5,6.5) -- (3,6);
 \draw [pattern=north west lines] (2.707,6) -- (2.707,6.5) -- (3.707,6.5) -- (3.707,6);
 \draw [fill] (1,7) -- (1.25,7) -- (1.25,6.75) -- (1,6.75);
 \draw [fill] (2.25,6.75) -- (2.25,7) -- (2.75,7) -- (2.75,6.75) -- (2.25,6.75);
 \draw [fill] (3.75,7) -- (4,7) -- (4,6.75) -- (3.75,6.75) -- (3.75,7);
 \draw [dotted] (4,7) -- (4.9,7);
 \draw [dotted] (4,6) -- (4.9,6);
 \draw [dotted] (3.75,6.5) -- (4.5,6.5);
 \draw [dotted] (4,6.75) -- (4.5,6.75);
 \node at (4.9,6.5) {\footnotesize $h_{b}$};
 \node at (4.3,6.835) {\footnotesize $h_{w}$};
 \node at (4.5,6.25) {\footnotesize $h_{l}$};
 \draw [arrows=->](4.9,6.65)--(4.9,7);
 \draw [arrows=->](4.9,6.35)--(4.9,6);
 \draw [arrows=->](4.5,6.68)--(4.5,6.5);
 \draw [arrows=->](4.5,5.82)--(4.5,6);
 \draw [arrows=->](4.3,6.6)--(4.3,6.7);
 \draw [arrows=->](4.3,7.1)--(4.3,7);
 \draw [arrows=->](2.5,8.5) -- (3.5,10);
 \node [above] at (3.2,9) {\footnotesize $R$};
 \draw [dotted] (2.707,6) -- (2.707,3.707);
 \draw [dotted] (3.707,6) -- (3.707,3.707);
 \draw [dotted] (4,6) -- (4,2.5);
 \draw [dotted] (1,6) -- (1,2.5);
 \node [above] at (2.5,5.2) {\footnotesize FRONT VIEW};
 \draw (2.5,3) circle [radius=2];
 \draw [dashed] (2.5,3) circle [radius=1.5];
 \draw (2.5,3) circle [radius=1];
 \draw (2.5,3) circle [radius=0.5];
 \draw [fill] (2.25,1.5) -- (2.25,1.75) -- (2.75,1.75) -- (2.75,1.5) -- (2.25,1.5);
 \draw [fill] (2.25,4.25) -- (2.25,4.5) -- (2.75,4.5) -- (2.75,4.25) -- (2.25,4.25);
 \draw [fill] (4,2.75) -- (4,3.25) -- (3.75,3.25) -- (3.75,2.75) -- (4,2.75);
 \draw [fill] (1.25,2.75) -- (1.25,3.25) -- (1,3.25) -- (1,2.75) -- (1.25,2.75);
 \draw [pattern=north west lines] (3.207,3.707) circle [radius=0.5];
 \draw [pattern=north west lines] (3.207,2.292) circle [radius=0.5];
 \draw [pattern=north west lines] (1.792,3.707) circle [radius=0.5];
 \draw [pattern=north west lines] (1.792,2.292) circle [radius=0.5];
 \draw [dotted] (2.5,1.5) -- (6,1.5);
 \draw [dotted] (2.5,4.5) -- (6,4.5);
 \draw [dotted] (3.207,4.207) -- (5.2,4.207);
 \draw [dotted] (3.207,3.207) -- (5.2,3.207);
 \node at (5.2,3.707) {\footnotesize $D_{l}$};
 \node at (6,3) {\footnotesize $D_{b}$};
 \draw [dotted] (2.25,4.25) -- (0.5,4.25);
 \draw [dotted] (2.25,4.5) -- (0.5,4.5);
 \draw [dotted] (2.25,4.5) -- (2.25,4.75);
 \draw [dotted] (2.75,4.5) -- (2.75,4.75);
 \node at (2.5,4.75) {\footnotesize $l_{w}$};
 \node at (0.5,4.375) {\footnotesize $w_{w}$};
 \draw [arrows=->](6,3.2)--(6,4.5);
 \draw [arrows=->](6,2.8)--(6,1.5);
 \draw [arrows=->](5.2,3.857)--(5.2,4.207);
 \draw [arrows=->](5.2,3.557)--(5.2,3.207);
 \draw [arrows=->](2.9,4.75) -- (2.75,4.75);
 \draw [arrows=->](2.1,4.75) -- (2.25,4.75);
 \draw [arrows=->](0.5,4.1) -- (0.5,4.25);
 \draw [arrows=->](0.5,4.65) -- (0.5,4.5);
 \node [above] at (2.5,0.1) {\footnotesize TOP VIEW};
 \end{tikzpicture}
 }%
\caption{This figure is to a scale of 1:30. The front and top views of the Li-Fi balloon have been shown. All dimensions are in centimetres (cm) as given in Table I. The balloon's aerial structure can be of any shape depending on the aerodynamical aspects. In the top view, the dashed circular boundary depicts the base of the balloon. There are four shaded regions on the base, on each of which a Li-Fi LED and an uplink receiver are placed. The space in the centre of the base is reserved for the LED drivers, constant power source and signal processing modem. For inter-balloon communication, at the corners of the base, four wireless nodes are placed (filled regions), with beamforming capability.}
\label{overview}
\end{figure}

\begin{table}
\caption{dimensions of the balloon}
\label{abc}
\begin{tabular}{ | m{3cm} | m{2cm}| m{2.5cm} |} 
\hline
Part of the Balloon & Symbol used & Dimension \\ 
\hline
Diameter of the base & $D_{b}$ & $90$cm \\ 
\hline
Height of the basket & $h_{b}$ & $45$cm  \\
\hline
Radius of the balloon sheet & $R$ & $70$cm \\
\hline
Diameter of the area for downlink and uplink components (Shaded area) & $D_{l}$ & $30$cm  \\
\hline
Height of the Downlink LED & $h_{l}$ & $15$cm  \\
\hline
Height of the wireless node(fixed near the roof of the basket) & $h_{w}$ & $10$cm  \\
\hline
Length of the wireless node & $l_{w}$ & $10$cm \\
\hline
Width of the wireless node & $w_{w}$ & $10$cm \\
\hline 
\end{tabular}
\end{table}
 
\subsection{The Base}
A Li-Fi balloon will contain all the necessary communication equipment and the power source. The construction of the balloon is as follows. It consists of a base (closed basket) onto which four downlink LEDs and uplink receivers are placed (shaded area). The shaded areas surround the empty area at the centre, where the power and signal processing modem is placed. Four wireless nodes (with beam forming) are placed in alternate corners of the base, for inter-balloon communication. 
The material used for the basket can be any strong waterproof and non-conducting fiber \cite{base}.  \\ \par
Distribution of weight should be uniform and homogenous to stabilize the balloon at a given height. So, a symmetric arrangement of components is desired. 

\subsection{The balloon sheet}
The frame of the balloon's arial sheet, can be of circular or hexagonal shape, appropriate to the aerodynamic balancing aspects \cite{sheet1}. The material used can be a thin rigid polymer as proposed in \cite{poly1} and \cite{poly2} or water resistant polymers as in \cite{poly3}. Light weight composites are preferred.   

\subsection{Components and Hardware}
The Li-Fi components used are shown in Table II.
\begin{table}
\caption{components used per balloon.}
\label{abc}
\begin{tabular}{ | m{2.5cm} | m{2.5cm}| m{1cm} | m{1cm} |} 
\hline
Component & Name & Quantity & Payload (max. value)  \\ 
\hline
 Luxspace DN561B & Philips downlink Li-Fi transmitter& $4$ & $0.4\times4=1.6$kg   \\ 
\hline
Philips modem board & Philips modem board & $1$ & $0.2$kg \\
\hline
LBRD1514-1& Xitanium $20$W LED power driver & $1$ & $0.3$kg \\
\hline
LBRD14016-3 & Uplink-IR Receiver & $4$ & $0.1\times4=0.4$kg  \\
\hline
DC Power supply & Rechargeable DC power supply & $1$ & $2$kg  \\
\hline
D-Link wireless access points & Wireless nodes A & $4$ & $4\times0.5=2$kg  \\
\hline
Motion and Position sensor & Sensor & $1$ (each) & $0.1$kg \\
\hline 
\end{tabular}
\end{table}
From Table II, we have the total payload per balloon as $6.6$kg. The dimensions in Table I suggest that the balloon will be closely packed. We now discuss the use of each component in Table II. The Philips \textit{Luxspace DN561B} is a Li-Fi downlink transmitter and \textit{LBRD14016-3} is an uplink infrared receiver. These are used for illumination and data access with various emission properties \cite{ph1}. In \cite{atc1}, these components have been used to characterize outage in an indoor Li-Fi environment. The \textit{LBRD1514-1} is a set of LED power drivers which are an integral part for intensity modulation (IM) purposes. The modem board acts as a central switch for all the LEDs and for IM. The modem board, drivers and the DC power supply can be appropriately fitted in the central space provided. The sensors will help in stabilizing the position and the inclination of the balloon. The wireless nodes will be useful to form an interconnected wireless balloon network.   

\section{The LiBNet}
A network of Li-Fi balloons is called LiBNet. Each balloon, using the wireless nodes, forms an interconnected wireless network. The balloons can be assumed to be arranged as a Poisson point process\footnote{This assumption gives the worst case co-channel interference in the network.}, as in \cite{coverage}, in both one and two dimensions which is shown in Fig. \ref{one} and Fig. \ref{two} respectively. For any arrangement, we consider all the balloons to be stabilized at a uniform height $h$ and have an effective\footnote{Effective, because each balloon has four downlink LEDs.} Lambertian emission order $m=\frac{-\ln(2)}{\ln(\cos(\theta_{h}))}$. $\theta_{h}$ is the effective half-power-semi-angle (HPSA) of the balloon. The trapped victim is assumed to be at a distance $z$ from the tagged balloon, which is nearest to the victim.      

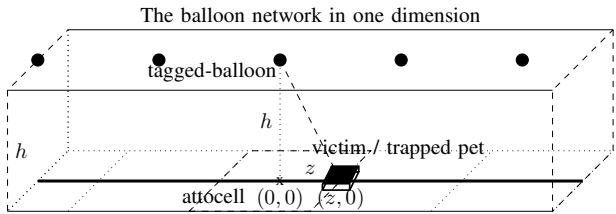
\begin{figure}[ht]
\centering
\resizebox{0.45\textwidth}{!}{%
 \begin{tikzpicture}
 \draw [dashed] (0,0) -- (0,2) -- (1,3);
 \node [right] at (0,1) {\normalsize $h$};
 \node [left] at (4.5,1.5) {\normalsize $h$};
 \node [above] at (5,0.5) {\normalsize $z$}; 
 \draw (1,3) -- (10,3);
 \draw [dashed] (10,3) -- (10,1) -- (9,0) -- (9,2) -- (10,3); 
 \draw (0,2) -- (9,2);
 \draw (0,0) -- (9,0);
 \draw [fill] (0.5,2.5) circle [radius=0.10];
 \draw [fill] (2.5,2.5) circle [radius=0.10];
 \draw [fill] (4.5,2.5) circle [radius=0.10];
 \draw [fill] (6.5,2.5) circle [radius=0.10];
 \draw [fill] (8.5,2.5) circle [radius=0.10];
 \draw [dotted] (1,3) -- (1,1) -- (0,0);
 \draw [dotted] (1,3) -- (1,1) -- (0,0);
 \draw [dotted] (1,1) -- (10,1);
 \draw [line width= 0.05cm] (0.5,0.5) -- (9.5,0.5);
 \node [above] at (5,3) {\normalsize The balloon network in one dimension };
 \node [below] at (4.5,0.5) {\normalsize $(0,0)$};
 \node [below] at (5.5,0.5) {\normalsize $(z,0)$};
 \node at (4.5,0.5) {\footnotesize x};
 \draw [thick] (5.2,0.35) -- (5.65,0.35) -- (5.8,0.65) -- (5.35,0.65) -- (5.2,0.35);
 \draw [fill] (5.2,0.45) -- (5.65,0.45) -- (5.8,0.75) -- (5.35,0.75) -- (5.2,0.45);
 \draw [thick] (5.2,0.35) -- (5.2,0.45);
 \draw [thick] (5.65,0.35) -- (5.65,0.45);
 \draw [thick] (5.8,0.65) -- (5.8,0.75);
 \draw [thick] (5.35,0.65) -- (5.35,0.75);
 \draw [dashed] (4.5,2.5) -- (5.5,0.5);
 \draw [dotted] (4.5,0.5) -- (4.5,2.5);
 \draw [dotted] (1,0) -- (2,1);
 \draw [dotted] (3,0) -- (4,1);
 \draw [dotted] (5,0) -- (6,1);
 \draw [dotted] (7,0) -- (8,1);
 \draw [dotted] (9,0) -- (10,1);
 \draw [dashed] (3,0) -- (5,0) -- (6,1) -- (4,1) -- (3,0);
 \node [above] at (3.4,0.001) {\normalsize attocell};
 \node [left] at (4.55,2.3) {\normalsize tagged-balloon};
 \node [above] at (6.45,0.75) {\normalsize victim / trapped pet};
 \end{tikzpicture}
 }%
\caption{\textit{(One dimension model)} This figure shows the one dimension balloon network. The balloons can be assumed to be arranged deterministically or as a Poisson point process $\psi$, stabilized at a height $h$ (black circular dots). The rectangular dotted regions on ground depict the attocells corresponding to each balloon above. The trapped victim (small cuboid) at $(z,0)$ (inside one of the attocell), receives data wirelessly from the tagged-balloon corresponding to the attocell in which he/she is located. Here, that attocell is highlighted as dash-dot. All other balloons become the co-channel interferers. Here, we assume that the victim can be trapped only along the thick line on ground. Every balloon is interconnected with each other wirelessly.}
\label{one}
\end{figure} 

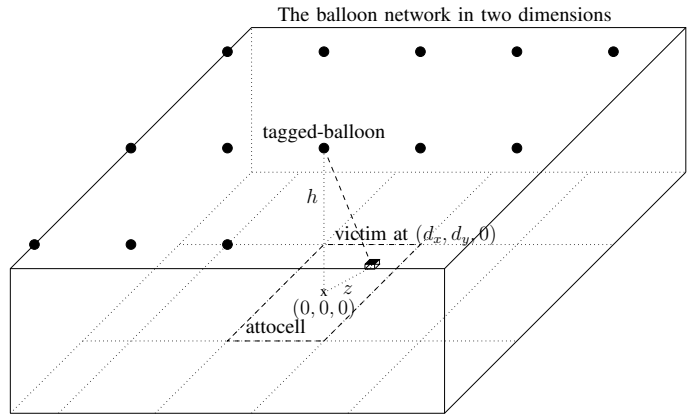
\begin{figure}[ht] 
\centering
 \resizebox{0.5\textwidth}{!}{%
 \begin{tikzpicture}
 \draw (0,0) -- (0,3) -- (5,8) -- (14,8) -- (14,5);
 \draw (14,5) -- (9,0) -- (0,0);
 \draw (14,8) -- (9,3) -- (0,3);
 \draw (9,0) -- (9,3);
 \draw [dotted] (0,0) -- (5,5) -- (14,5);
 \draw [dotted] (5,5) -- (5,8);
 
 \draw [fill] (0.5,3.5) circle [radius=0.10];
 
 \draw [fill] (2.5,3.5) circle [radius=0.10];
 
 \draw [fill] (4.5,3.5) circle [radius=0.10];
 
 
 
 \draw [fill] (2.5,5.5) circle [radius=0.10];
 \node at (5.5,1.8) {\large attocell};
 
 \draw [fill] (4.5,5.5) circle [radius=0.10];
 
 \draw [fill] (6.5,5.5) circle [radius=0.10];
 \node [above] at (6.5,5.5) {\large tagged-balloon};
 
 \draw [fill] (8.5,5.5) circle [radius=0.10];
 
 \draw [fill] (10.5,5.5) circle [radius=0.10];
 
 \draw [fill] (4.5,7.5) circle [radius=0.10];
 
 \draw [fill] (6.5,7.5) circle [radius=0.10];
 
 \draw [fill] (8.5,7.5) circle [radius=0.10];
 
 \draw [fill] (10.5,7.5) circle [radius=0.10];
 
 \draw [fill] (12.5,7.5) circle [radius=0.10]; 
 
 \draw [dotted] (1,0) -- (6,5);
 \draw [dotted] (3,0) -- (8,5);
 \draw [dotted] (5,0) -- (10,5);
 \draw [dotted] (7,0) -- (12,5);
 \draw [dotted] (1.5,1.5) -- (10.5,1.5);
 \draw [dotted] (3.5,3.5) -- (12.5,3.5);
 \draw [dotted] (6.5,2.5) -- (6.5,5.5);
 \node [align=center] at (6.5,2.5) {\small x};
 \node [below] at (6.5,2.5) {\large $(0,0,0)$};
 \node [left] at (6.5,4.5) {\large $h$}; 
 \draw (7.35,2.99) -- (7.55,2.99) -- (7.65,3.09) -- (7.45,3.09) -- (7.35,2.99);
 \draw [fill] (7.35,3.09) -- (7.55,3.09) -- (7.65,3.19) -- (7.45,3.19) -- (7.35,3.09);
 \draw (7.35,2.99) -- (7.35,3.09);
 \draw (7.55,2.99) -- (7.55,3.09);
 \draw (7.65,3.09) -- (7.65,3.19);
 \draw (7.45,3.09) -- (7.45,3.19);
 \draw [dotted] (6.5,2.5) -- (7.5,3.04);
\node [below] at (8.4,4.04) {\large victim at $(d_{x},d_{y},0)$};
 \node [below] at (7,2.77) {\large $z$};
 \draw [dashed] (6.5,5.5) -- (7.5,3.04); 
 \draw [dashed] (6.5,1.5) -- (8.5,3.5) -- (6.5,3.5) -- (4.5,1.5) -- (6.5,1.5);
 \node [above] at (9,8) {\large The balloon network in two dimensions}; 
 
 \end{tikzpicture}
}%
\caption{\textit{(Two dimension model)} This figure shows the two dimension balloon network. The balloons can be assumed to be arranged deterministically or as a Poisson point process $\psi$, stabilized at a height $h$ (black circular dots). The rectangular dotted regions on ground depict the attocells corresponding to each balloon above. The trapped victim (small cuboid) at $(d_{x},d_{y},0)$ (inside one of the attocell), receives data wirelessly from the tagged-balloon corresponding to the attocell in which it is located. Here, that attocell is highlighted as dash-dot. All other balloons become the co-channel interferers. Here we assume that the victim can be trapped anywhere on the ground plane. Every balloon is interconnected with each other wirelessly.}
\label{two}
\end{figure} 

\subsection{How does the LiBNet work?}
Each balloon with the necessary Li-Fi equipment will cover a certain affected region via illumination. Uplink signatures from the trapped lives, can be received and downlink broadcast can be done. The sensing of uplink signatures can be done efficiently by using machine learning as in \cite{atckali} for cognitive radio networks. All the balloons will be interconnected wirelessly and linked to a fusion centre (FC) which is in-turn connected to the Internet. The uplink and downlink data will be routed through the network and through the FC as a gateway to the Internet. This model can be assumed similar to a wireless sensor network (WSN) in a cooperative sensing environment \cite{wsn}. Now, various routing protocols can be used or devised \cite{atc2}, \cite{wsn2} to route the data efficiently. So, such a network of Li-Fi balloons can be used anywhere in the affected region. Their interconnected network can be used to track the lives of many and save them immediately.

\subsection{The mean co-channel interference}
Let all the points in the given balloon network be arranged as a Poisson point process $\psi$, in both one and two dimensions. Given an origin, let the balloons be at a distance $x$ from the origin. The user photodiode (PD) is at a distance of $z$ from the origin. Let the tagged balloon be at origin i.e. $x=0$. Now, assuming a wavelength reuse factor of unity, all other balloons become co-channel interferers. We now derive the mean co-channel interference in the Poisson network of balloons with the assumptions in \cite{chen}. For a distance $D_{x}$ between the PD and the balloon, the Signal-to-Interference-plus-Noise-Ratio (SINR) for such a network is given as 
\begin{equation}
\gamma(z) = \frac{(z^{2}+h^{2})^{-m-3}\rho(D_{0})}{ \sum_{x\in\psi \setminus 0}( (x + z)^{2} + h^{2} )^{-m-3}\rho(D_{x})  + \Omega},
\label{eqn:sinr}  
\end{equation}
where, $\rho(D_{x})$ is the field-of-view (FOV) $\theta_{f}$ constraint function of the PD used by the trapped victim, which is defined as 
\begin{equation*}
\rho(D_{x}) = \left\{
               \begin{array}{ll}
                 1,& |D_{x}|\leq h\tan(\theta_{f}), \\
                 0,& |D_{x}|> h\tan(\theta_{f}).
              \end{array}
              \right.
\end{equation*}

The co-channel interference $\mathcal{I}_{x}$ is
\[\mathcal{I}_{x} = \sum_{x\in\psi \setminus 0}f(x),\]
where $f(x)$ is given as 
\[ ((x + z)^{2} + h^{2} )^{-m-3}\rho(D_{x}).  \]    
Hence, we state the theorem for mean interference below. 
\begin{theorem}
Consider a photodiode, with FOV $\theta_{f}$ radians, situated at a distance $z$ (inside an attocell) from the origin, in a Poisson distributed balloons of Li-Fi LEDs, emitting light with an effective Lambertian emission order $m$, installed at a height $h$. Then, for a wavelength reuse factor of unity, the mean co-channel interference caused by the Poisson distributed co-channel interferers of intensity $\lambda(x)$ at the photodiode is
\begin{align*}
\mathbb{E}(\mathcal{I}_{x})=\int_{\mathbb{S}}\lambda(x)f(x)\d x,
\end{align*}
where, $\mathbb{E}(.)$ is the expectation operator over all possible random arrangements in the Poisson point process and $\mathbb{S}$ is the support of the integration. 
\label{theorem1}
\end{theorem}
\begin{proof}
The proof is provided in Appendix \ref{app:theorem1}.
\end{proof}
We now extend Thm. \ref{theorem1} to one and two dimension LiBNets in the following Lem. \ref{lem1} and \ref{lem2} respectively.

\begin{lemma}
The mean co-channel interference in a one dimension LiBNet with homogenously distributed Poisson point balloons is given as
\begin{align*} 
\mathbb{E}(\mathcal{I}_{x})&=\lambda\bigg(h^{1-2\beta}\tan(\theta_{f}){}_{2}F_{1}\bigg(0.5,\beta;1.5;-\tan^{2}(\theta_{f})\bigg) \nonumber\\
&\ \ \ \ \ - zh^{-2\beta}{}_{2}F_{1}\bigg(0.5,\beta;1.5;-\frac{z^{2}}{h^{2}}\bigg)\bigg),
\end{align*}
where, ${}_{2}F_{1}(.;.;.)$ is the ordinary hypergeometric function and $\beta=m+3$. 
\label{lem1}
\end{lemma}
\begin{proof}
From Thm. \ref{theorem1}, with the support $\mathbb{S}=[z,h\tan(\theta_{f})]$ we have  
\begin{align}
\mathbb{E}(\mathcal{I}_{x})&=\int_{z}^{h\tan(\theta_{f})}\lambda(x)f(x)\d x,
\label{eqn:lem1p1}
\end{align}
because, the interferers are present only after a distance $z$ limited by the FOV of the PD. Now, for a one dimension network
\[ f(x) = \frac{1}{(x^{2}+h^{2})^{\beta}} \]
and $\lambda(x)=\lambda$ for a homogenous point process. Hence, integrating the same in \eqref{eqn:lem1p1} derives the result.
\end{proof}

\begin{lemma}
The mean co-channel interference in a two dimension LiBNet homogenously distributed Poisson point balloons is given as
\begin{align*} 
\mathbb{E}(\mathcal{I}_{x})&=\lambda\bigg(\frac{\pi(h^{2}+z^{2})^{1-\beta}}{\beta-1} - \frac{\pi h^{2-2\beta}\cos^{2\beta-2}(\theta_{f})}{\beta-1}\bigg).
\end{align*}
where, $\beta=m+3$.
\label{lem2}
\end{lemma}
\begin{proof}
For a two dimension homogenous Poisson point process of balloons we have $\lambda(x)=\lambda$ and $f(x,y)=\frac{1}{(x^{2}+y^{2}+h^{2})}$. So, from Thm. \ref{theorem1}, we can write
\[   \mathbb{E}(\mathcal{I}_{x}) = \int\int_{\mathbb{S}} \lambda \frac{1}{(x^{2}+y^{2}+h^{2})} \d x \d y. \]
Converting to polar coordinates $(r,\theta)$, we have
\[ \mathbb{E}(\mathcal{I}_{x}) = 2\pi \int_{z}^{h\tan(\theta_{f})} \frac{1}{(r^{2}+h^{2})} \d r. \]
Integrating further, derives the result.
\end{proof}
 
\section{Conclusion}
In this paper, a novel idea of using balloons for Li-Fi has been proposed. A disaster or a natural calamity has been taken as a use case. An overview of the physical design of the balloon has been given. Given the payload, the installation of the communication components has been proposed using the Philips Li-Fi equipments. Symmetry of installation has been kept in mind. Further, the functioning of a network of such Li-Fi balloons has been discussed and coined an acronym of LiBNet. The mean co-channel interference in both one and two dimension LiBNets has been derived, assuming the balloons are arranged as a homogenous Poisson point process.
Given the problems of contemporary interest, the possibilities of future work are extensive. Li-Fi, being a safe and high speed alternative to RF emergency communication set-ups, this concept, if implemented can save the lives of many. Tracking and positioning of people / pets becomes easier. But, Li-Fi, requiring too much line of sight, will have to co-exist with blockages. Hence an extension to a blockage model is desired. Also, as a future work, this concept can be further standardised and practically implemented.    
  
\bibliographystyle{IEEEtran}
\bibliography{IEEEfull,Reference1.bib}

\begin{appendices}
\section{}
\section*{Proof of Mean Interference}
\begin{proof}

The Laplace transform of the Interference $\mathcal{I}_{x}$ is given as

\begin{align}
\mathcal{L}_{\mathcal{I}_{x}}(s)&=\mathbb{E}_{\psi}(e^{-s\mathcal{I}_{x}}), \nonumber\\
&=\mathbb{E}_{\psi}(e^{-s\sum_{x \in \psi \setminus 0}f(x)}), \nonumber\\
&=\mathbb{E}_{\psi}\bigg(\prod_{x \in \psi \setminus 0 } e^{-sf(x)}\bigg),\nonumber\\
&\stackrel{(a)}{=}e^{-\lambda(x)\int_{\mathbb{R}}(1-e^{-sf(x)})\d x},
\label{eqn:first}
\end{align}
where, (a) follows from the definition of expectation operator. 
Now, taking the $n^{th}$ order derivative on L.H.S of \eqref{eqn:first} w.r.t $s$, we have 
\[(-1)^{n}\frac{\partial^{n}}{\partial s^{n}}(\mathbb{E}(e^{-s\mathcal{I}_{x}}))=\mathbb{E}(\mathcal{I}^{n}_{x}e^{-s\mathcal{I}_{x}}).\] 
Using the result in \eqref{eqn:first}, we can write for $n=1$ as
\begin{align}
\mathbb{E}(\mathcal{I}_{x}e^{-s\mathcal{I}_{x}})&= (-1)e^{-\lambda(x)\int_{\mathbb{R}}(1-e^{-sf(x)})\d x}\nonumber\\
&\ \ \ \ \ \ (-1)\big(\lambda(x)\int_{\mathbb{R}}(-1)(-f(x))\d x\big).
\label{eqn:second}
\end{align}
Substituting $s=0$ in \eqref{eqn:second} we have 
\[ \mathbb{E}(\mathcal{I}_{x})=\int_{\mathbb{R}}\lambda(x)f(x)\d x,   \]
proving the theorem. 
\end{proof}

\label{app:theorem1}  
\end{appendices}

\end{document}